\documentclass[a4paper,10pt]{article}

\usepackage{amsthm,amsmath,amscd,amssymb,amsfonts,pstricks,pst-node,graphics}

\def \diag {\, \mbox{diag}}

\def \ad {\, \mbox{ad}}

\def \psh {{\Omega}}

\def\bbbz{{\mathbb Z}}

\def\bbbd{{\mathbb D}}

\def\cp1{{\mathbb C\mathbb P}^1}

\def\cU{{\cal U}}

\def\cW{{\cal W}}

\def\cM{{\cal M}}

\def\cO{{\cal O}}
\def\cF{{\cal F}}
\def\cS{{\cal S}}
\def\cT{{\cal T}}

\def\T{{T}}

\def\bu{{\mathbf u}}

\def\res{{\mbox{\rm res\,}}}
\def\bpi{{\bf \pi}}

\newtheorem{Def}{Definition}

\newtheorem{Thm}{Theorem}
\newtheorem{Pro}{Proposition}
\newtheorem{Lem}{Lemma}
\newtheorem{Cor}{Corollary}

\def\acS{{\, \xrightarrow{\cS}\, }}
\begin{document}
\title{Darboux transformation with Dihedral reduction group}
\author{ Alexander V. Mikhailov$^{\T}$, Georgios Papamikos$^{\dagger}$ and Jing Ping Wang$ ^\dagger $
\\
$\dagger$ School of Mathematics, Statistics \& Actuarial Science, University of Kent, UK \\
$\T$ Applied Mathematics Department, University of Leeds, UK
}
\date{}
\maketitle
\begin{abstract}
We construct the Darboux transformation with Dihedral reduction group for the $2$--dimensional generalisation of the periodic Volterra
lattice.
The resulting B{\"a}cklund transformation can be viewed as a nonevolutionary integrable differential difference equation. We also find 
its generalised symmetry and the Lax representation for this symmetry. Using formal diagonalisation of the Darboux matrix we obtain local conservation laws
of the system.
\end{abstract}

\section{Introduction}

The connection among integrable partial differential
equations, 
Darboux transformations of the corresponding Lax operators and partial difference 
equations and the corresponding B\"acklund transformations is well established (see, for instance, 
\cite{mat91, mr1908706, BobSuris, kmw13}).
Namely, the Darboux 
transformations serve as Lax representations for integrable difference equations, 
while the B\"acklund transformations are symmetries of these difference equations 
and are integrable differential-difference equations in their own right. 

The motivation of this research is to extend the reduction 
groups \cite{mik79,mik80,mik81} of Lax representations of integrable partial differential
equations to integrable difference equations via Darboux transformations, and ultimately to describe all elementary Darboux
transformations corresponding to affine Lie algebras and automorphic Lie
algebras with finite reduction groups. Recently, the authors of \cite{smx, krm13} have completed a comprehensive study for 
the Lax operators of the nonlinear Schr{\"o}dinger equation type and they derived some new discrete equations and new maps.
Here we go beyond $sl_2$ Lax representations and take outer automorphisms into account.

The concept of reduction groups for Lax representations of partial differential equations was first introduced in \cite{mik79,mik80,mik81}.
Later, a new class of infinite dimensional quasi-graded Lie algebra called automorphic Lie algebras based on a reduction group was studied in \cite{LM05}.
The reduction groups have been used to classify  Lax representations and their corresponding integrable equations \cite{LM04, bury}.
The study of automorphic Lie algebras has become a research topic in its own right \cite{ls10,kls14}.

In this paper, we construct the Darboux transformation for the $2$--dimensional generalisation
of the Volterra lattice \cite{mik79}:
\begin{eqnarray}\label{2+1}
\left\{ \begin{array}{l}\phi^{(i)}_{t}= \theta^{(i)}_{x} + \theta^{(i)} \phi^{(i)}_{x}
-e^{2 \phi^{(i-1)}}+ e^{2 \phi^{(i+1)}}, \\
\theta^{(i+1)}-\theta^{(i)}+ \phi^{(i+1)}_{x}+\phi^{(i)}_{x}=0 , \end{array}\right.\quad
\phi^{(i+n)}=\phi^{(i)}, \quad \theta^{(i+n)}=\theta^{(i)},\quad \sum_{i=1}^n \phi^{(i)}=\sum_{i=1}^n\theta^{(i)}=0 .
\end{eqnarray}
This system can be viewed as a discretisation of
the Kadomtsev-Petviashvili equation. Indeed, in the limit $n\to \infty$,
\[\phi^{(i)}(x,t)=h^2 u(\xi,\eta,\tau),\quad h=n^{-1}, \]
\[ \tau=h^3 t, \quad \xi=i h+4ht,\quad \eta=h^2 x, \]
 system (\ref{2+1}) goes to
\[ u_\tau=\frac{2}{3}u_{\xi\xi\xi}+8uu_\xi-2D_\xi^{-1}u_{\eta\eta}+O(h^2).\]
The exact solutions of system (\ref{2+1}) have much in common with $2+1$-dimensional integrable equations, which have recently been investigated by
Bury and Mikhailov \cite{bury,bm}. For fixed period $n$, it is a bi-Hamiltonian system.
When $n=3$, its recursion operator and bi-Hamiltonian structure are explicitly constructed from its Lax representation in 
\cite{wang09}.

The Lax representation of (\ref{2+1}) is invariant under the  the dihedral reduction group $\bbbd_n$ \cite{mik79,mik81,LM05, LM04,bury}
with both inner and outer Lie algebra automorphisms.  To the best of our knowledge, this is the first example in which symmetries of a Darboux transformation are also generated by
outer automorphisms. We rigorously prove that a B{\"a}cklund transformation for the $2$--dimensional generalisation of the Volterra lattice (\ref{2+1}) is
\begin{eqnarray}
(\cS\!\!-\!\!1) \phi_x^{(k)}=-\frac{n}{\mu^2} (\psh-1)^2 \left(a^{(k-1)} b^{(k)} p^{(k-1)}\exp(\phi^{(k-1)})\right),\label{flowx0}
\end{eqnarray} 
where $\cS$ is a shift operator and $\psh$ is an $n$ periodic shift operator acting on the upper index, and $a^{(k)}$, $ b^{(k)}$ and  $p^{(k)}$ are defined
in terms of $\phi^{(i)}$, $i,k=1,\cdots, n-1$ by
\begin{eqnarray*}
&& p^{(i)}=\mu \exp\{(\sum_{k=i}^{n-1}-\sum_{k=1}^{n-1}\frac{k}{n}) (\phi_1^{(k)}-\phi^{(k)})\};\\
&&{a^{(i)}}^2=\frac{({p^{(i)}}^2-1) \mu^{2(n-i)} \prod_{l=1}^{i-1} {p^{(l)}}^2 } {\mu^{2 n}-\prod_{l=1}^{n-1} {p^{(l)}}^2 };\qquad
b^{(i)}=\frac{ a^{(i)} \prod_{l=i}^{n-1} {p^{(l)}}^2 ({p^{(n)}}^2-1)}{n \mu^{2(n-i)-1} p^{(i)} }.
\end{eqnarray*}
Moreover, for the nonevolutionary integrable differential difference equation (\ref{flowx0}), we find its generalise symmetry
$$\phi_{\tau}^{(k)}=-\frac{n }{\mu <b,\ a_{-1}>} (\psh-1) \left( b^{(k)} a_{-1}^{(k)} \right),$$
where $<b,\ a_{-1}>=\sum_{i=1}^{n} b^{(i)}a_{-1}^{(i)}$ 
and conserved densities
\begin{eqnarray*}
&& \rho_0=\log <b,\ a_{-1}> ; \\
&& \rho_1=-\frac{1}{\mu <b_1 ,\,a><b,\, a_{-1}> }\sum_{k,i=1}^n
\Gamma_{k-q}b_1^{(k)}a^{(k)}b^{(i)}a_{-1}^{(i)}, \quad \Gamma_l=(l-1)\!\!\! \mod\! n-\frac{n-1}{2}.  
\end{eqnarray*}

The arrangement of this paper is as follows: In Section \ref{Sec21},  we give basic definitions related to the Lax--Darboux scheme such as Darboux transformation, Darboux map and
B{\"a}cklund transformation. Meanwhile, we also fix notations. In  Section \ref{Sec22}, we parametrise the rank $1$ matrix with simple poles invariant
under the dihedral group $\bbbd_n$.
Then in Section \ref{Sec23} we provide a rigorous proof for the rank $1$ Darboux transformation of the $2$--dimensional generalisation
of Volterra lattice (\ref{2+1}) and give the corresponding B{\"a}cklund transformation, which can be viewed as an integrable nonevolutionary differential difference
equation. In Section \ref{Sec3}, we construct a local generalised symmetry for this nonevolutionary equation using the Darboux matrix. We also provide 
the Lax representation for this symmetry flow viewed as an evolutionary integrable equation. 
In Section \ref{Sec4}, we transform the Lax pair for this symmetry flow into a formal block-diagonal form to get its conservation laws. 
The nonevolutionary integrable equation and its symmetry share the same conserved densities. Finally, we give a brief conclusion on what we have done in this paper and 
a short discussion about open problems. To improve the readability, we put the technical results required in the Appendix.

\section{Generalised Volterra lattice and related structures}\label{Sec21}
In this section, we discuss the known Lax representation for the $2$--dimensional generalisation
of the Volterra lattice (\ref{2+1}) and its reduction group.  Meanwhile we fix notations and introduce the definitions related to the
Lax--Darboux scheme such as a Darboux transformation, Darboux map and B{\"a}cklund transformation.

Consider the following linear differential operators with matrix coefficients \cite{mik79,mik81}
\begin{eqnarray}\label{laxva}
\begin{array}{ll}
   ×L^{(1)}=D_{t_1}-V^{(1)},& V^{(1)}= \lambda\mathbf{u}
\Delta-\lambda^{-1}\Delta^{-1} \mathbf{u},\\
L^{(2)}=D_{t_2}-V^{(2)},& V^{(2)}= \lambda
^2\mathbf{u}\Delta\mathbf{u}\Delta+\lambda\mathbf{a}
\Delta-\lambda^{-1}
\Delta^{-1}\mathbf{a}-\lambda^{-2}  \Delta^{-1} \mathbf{u}\Delta^{-1}\mathbf{u} ,
  \end{array}
\end{eqnarray}
where $t_1=x,\ t_2=t$,   $\mathbf{u}=\mbox{diag}(\exp (\phi^{(i)})),$ $\mathbf{a}=\diag
(\theta_i \exp ( \phi^{(i)}))$ are diagonal matrices and $\Delta_{ij}=\delta_{i,j-1}$
whose indexes are count modulo $n$, i.e., 
\begin{equation}
 \label{Delta}
 \Delta =\begin{pmatrix}
0& 1 & 0 & \dots & 0 \\
0 & 0 & 1 & \dots & 0 \\
0& 0& 0& \dots&1 \\
1 & 0&0 &  \dots& 0
\end{pmatrix}.
\end{equation}
This convention to count indexes modulo $n$ will be used throughout the whole paper, thus $\phi^{(i+n)}=\phi^{(i)}$ and $\theta^{(i+n)}=\theta^{(i)}$.
We introduce an $n$ periodic shift operator as follows:
$$\psh(\phi^{(i)})=\phi^{(i+1)} \quad \mbox{for $i=1, \cdots, n-1$ and}\quad \psh(\phi^{(n)})=\phi^{(1)}.$$

We set $\sum_{i=1}^n \phi^{(i)}=\sum_{i=1}^n\theta^{(i)}=0$ without losing generality.
The condition of commutativity of these operators
\begin{equation}\label{zeroc}
 [L^{(1)},L^{(2)}]=D_{t}(V^{(1)})-D_x(V^{(2)})+[V^{(1)},V^{(2)}]=0
\end{equation}
leads to the $2$--dimensional generalisation of the Volterra lattice (\ref{2+1}). 
This is often called a zero curvature representation or Lax representation of equation (\ref{2+1}).
These two operators, $ L^{(1)}$ and $L^{(2)}$, are conventionally called the Lax pair.
The commutativity of operators  can be seen as a compatibility 
 condition for the linear problems
\begin{equation}\label{lax0}
 D_x(\Psi)=V^{(1)}(\bu; \lambda)\Psi,\qquad D_{t}(\Psi)=V^{(2)}(\bu; \lambda)\Psi,
\end{equation}
i.e. the condition for the existence of a common fundamental solution  $\Psi$ ($\det \Psi\ne 0$).

The operators $L^{(i)}$ are invariant under the group of automorphisms generated by following two transformations:
\begin{equation}\label{redgroup}
s: L^{(i)}(\lambda)\mapsto Q L^{(i)}(\lambda \omega) Q^{-1}\quad \mbox{and} \quad r:
L^{(i)}(\lambda)\mapsto -L^{(i)\T}(\lambda^{-1}),\qquad Q=\diag (\omega^i),\
\omega=\exp\frac{2\pi {\rm i}}{n} .
\end{equation}
These two transformations satisfy
$$s^n = r^2 = id,\quad rsr = s^{-1} $$
and therefore generate the dihedral group $\bbbd_n$ \cite{mik79,mik81,LM05,LM04, bury}.
Note that the transformation $r$ is an outer automorphism of the Lie algebra $sl(n)$ over the Laurent polynomial ring $\mathbb{C}[\lambda,\lambda^{-1}]$.

A Darboux transformation is a linear map acting on a fundamental solution 
\begin{equation}\label{SPsi}
\Psi\mapsto \overline{\Psi}=M_{\mu}\Psi,\qquad \det\,M_{\mu}\ne 0
\end{equation}
such that the matrix function $\overline{\Psi}$ is a fundamental solution of the linear problems 
\begin{equation}\label{lax1}
 D_x(\overline{\Psi})=V^{(1)}(\overline{\bu}; \lambda)\overline{\Psi},\qquad D_{t}(\overline{\Psi})=V^{(2)}(\overline{\bu}; \lambda)\overline{\Psi}
\end{equation}
with new ``potentials''  $\overline{\bu}$. 
The matrix $M_{\mu}$ is often called the Darboux matrix.
From the compatibility of (\ref{SPsi}) and (\ref{lax1}) it follows that 
\begin{eqnarray}
&& D_{t_i}(M_{\mu})=V^{(i)}(\overline{\bu}; \lambda)M_{\mu}-M_{\mu} V^{(i)}(\bu; \lambda), \quad i=1,2.\label{backlundx}
\end{eqnarray}
Equations (\ref{backlundx}) are differential equations which relate two solutions  $\bu$ 
and $\bar\bu$ of  (\ref{2+1}). In the literature they are also often called {\sl B\"acklund transformations}.

A Darboux transformation maps one compatible system (\ref{lax0}) into another one (\ref{lax1}). It defines a 
Darboux map $\cS: \bu\mapsto\overline{\bu}$.
The map (\ref{SPsi}) is invertible ($\det\,M_{\mu}\ne 0$) and it can be iterated
\[
  \cdots\underline{\Psi}\acS \Psi\acS\overline{\Psi}
  \acS\overline{\overline{\Psi}}\acS\cdots\, .
\]
It suggests notations
$$
 \ldots \Psi_{-1}=\underline{\Psi},\ \Psi_0=\Psi,\ \Psi_1=\overline{\Psi},\ \Psi_2=\overline{\overline{\Psi}}, \ldots,
$$
$$ \ldots \bu_{-1}=\underline{\bu},\ \bu_0=\bu,\ \bu_1=\overline{\bu},\ \bu_2=\overline{\overline{\bu}}, \ldots\, .$$
With a vertex $k$  of the one dimensional lattice $\bbbz$ we associate variables $\Psi_k$ and $ \bu_k$; 
with the edges joining vertices $k$ and $k+1$  we associate the matrix $\cS^k (M_{\mu})$. In these notations
the Darboux maps $\cS$ and $\cS^{-1}$ increase and decrease the subscript index by one, and therefore we shall call it  a $\cS$--shift, or shift operator $\cS$.
In what follows we often shall omit zero in the subscript index and write  $\bu$ instead of $\bu_0$. The map $\cS$ is an automorphism of the Lax structure (\ref{lax0})
and a discrete symmetry of system (\ref{2+1}) .

In these notations the resulting B\"acklund transformations from (\ref{backlundx}) are integrable differential difference equations
and (\ref{backlundx}) are their Lax-Darboux representations. 

A Darboux transformation with matrix $M_\mu$ with a parameter $\mu$ results in the $\cS$ shift. If we also consider a Darboux transformation $M_\nu$ 
with a different choice of the parameter, then the corresponding shift (automorphism of the Lax structure (\ref{lax0})) we denote $\cT$. The compatibility 
condition similar to (\ref{backlundx}) reads
\begin{eqnarray}
&& D_{t_i}(M_{\nu})=\cT(V^{(i)})M_{\nu}-M_{\nu} V^{(i)}, \quad i=1,2.\label{backlundt}.
\end{eqnarray}
These shifts act on $\bbbz^2$ lattice and with the vertex $(n,m)$ we
associate variables $\bu_{n,m}$, so that $\cS (\bu_{n,m})=\bu_{n+1,m},\ \cT (\bu_{n,m})=\bu_{n,m+1}$. Commutativity of the shifts $\cS$ and $\cT$ leads to a 
system of partial-difference equations 
\[
 \cT(M_\mu)M_\nu-\cS(M_\nu)M_\mu=0.
\]
Differential difference equations (\ref{backlundx}), (\ref{backlundt}) are the symmetries of this partial-difference equation. Indeed
\[
 D_{t_i}(\cT(M_\mu)M_\nu-\cS(M_\nu)M_\mu)=\cT\cS(V^{(i)})(\cT(M_\mu)M_\nu-\cS(M_\nu)M_\mu)-(\cT(M_\mu)M_\nu-\cS(M_\nu)M_\mu)V^{(i)}=0.
\]

The aim of this paper is to construct a Darboux matrix $M_\mu$ for the operators
$L^{(1)}$ and the corresponding integrable B{\"a}cklund transformation chain. We also show how to construct local symmetries for this chain
and the corresponding Lax operators.

\section{Parametrisation of $\bbbd_n$ invariant matrices $M_{\mu}(\lambda)$ with simple poles}\label{Sec22}
In this section, we show that a rational in $\lambda$ matrix which is invariant under the dihedral group $\bbbd_n$ and has simple poles at generic orbit
$\{\mu \omega^i,i=0,1, \cdots , n-1\}$ can be completely parametrised by $n-1$ variables.

We assume that the Darboux matrix  $M_\mu$  inherits the same reduction group (\ref{redgroup}). 
Namely, it satisfies
\begin{eqnarray}
&& M_\mu(\lambda)^{-1}=M_\mu^{\T}(\frac{1}{\lambda}), \qquad Q M_\mu(\lambda
\omega) Q^{-1}=M_\mu(\lambda)\label{InvM}
\end{eqnarray}
and has simple poles at $\mu \omega^i$, $k=1,2, \cdots, n (\mu \omega^i\neq 0, \pm1)$.
Here $\T$ is the transpose of a matrix.
We thus take it of the following form:
\begin{eqnarray}
&& M_\mu(\lambda)=C+\sum_{i=0}^{n-1} \frac{Q^i A Q^{-i}}{\lambda
\omega^i-\mu},\label{M}
\end{eqnarray}
where $C$ and $A$ are $n\times n$ matrices.
We now study the constraints on the Darboux matrix  $M_\mu$ following from the condition (\ref{InvM}). 

It follows from the second identity in (\ref{InvM}) that matrix $C$ satisfies $QCQ^{-1}=C$ and thus $C$ 
is diagonal. We write $C=\diag (p^{(i)})$.

\begin{Pro}\label{prop1}
Let $I_n$ denote the identity $n\times n$ matrix. The matrix $M_\mu(\lambda)$
given by (\ref{M}) is invariant under group $\bbbd_n$ if and only if the matrix
$A$ 
and diagonal matrix $C$ satisfy the relations:
\begin{eqnarray}
&&C^2-\frac{1}{\mu}\sum_{i=0}^{n-1} Q^{-i} A^{\T} C Q^{i}=I_n; \qquad
C A +\mu \sum_{j=0}^{n-1}  \frac{Q^{-j} A^{\T} Q^{j} A }{\omega^{j}-\mu^2}=0.\label{cond}
\end{eqnarray}
\end{Pro}
\begin{proof} To prove the statement, we need to check the identities in (\ref{InvM}). It is obvious that the given $M_\mu(\lambda)$ 
satisfies $Q M_\mu(\lambda \omega) Q^{-1}=M_\mu(\lambda)$.  We now compute
\begin{eqnarray*}
&&M_\mu^{\T}(\frac{1}{\lambda}) M_\mu(\lambda)=\left(C+\sum_{j=0}^{n-1}
\frac{Q^{-j} A^{\T} Q^{j}}{\frac{\omega^j}{\lambda} -\mu}\right) 
\left(C+\sum_{i=0}^{n-1} \frac{Q^i A Q^{-i}}{\lambda \omega^i-\mu}\right)\\
&&=C^2+\sum_{i=0}^{n-1} \left(\frac{Q^{-i} A^{\T} C Q^{i}}{\frac{\omega^i}{\lambda} -\mu}+ \frac{Q^i C A Q^{-i}}{\lambda \omega^i-\mu}\right)
+\sum_{i,j=0}^{n-1} \frac{Q^{-j} A^{\T} Q^{i+j} A Q^{-i}}{(\frac{\omega^j}{\lambda} -\mu) (\lambda \omega^i-\mu)}\\
&&=C^2-\frac{1}{\mu}\sum_{i=0}^{n-1} Q^{-i} A^{\T} C Q^{i}-\frac{1}{\mu}\sum_{j=0}^{n-1} \frac{\omega^jQ^{-j} A^{\T} C Q^{j}}{\mu \lambda-\omega^j} 
+\sum_{i=0}^{n-1} \frac{Q^i C A Q^{-i}}{\lambda \omega^i-\mu}\\
&&
+\sum_{i,j=0}^{n-1} \frac{Q^{-j} A^{\T} Q^{i+j} A Q^{-i}}{\omega^{i+j}-\mu^2}\left(\frac{\mu}{\lambda \omega^i-\mu}
-\frac{\omega^j}{\lambda \mu -\omega^j}\right),
\end{eqnarray*}
which is the identity matrix according to (\ref{InvM}). 
Taking the limit $\lambda\to\infty$ we obtain the first equation in (\ref{cond}). The singular part (residues) vanishes if
\begin{eqnarray}
&&C A +\mu \sum_{j=0}^{n-1}  \frac{Q^{-i-j} A^{\T} Q^{i+j} A }{\omega^{i+j}-\mu^2}=0, \quad \mbox{ for $i=0,1,\cdots, n-1$};\label{con2}\\
&&\frac{1}{\mu}  A^{\T} C +\sum_{i=0}^{n-1} \frac{A^{\T} Q^{i+j} A Q^{-i-j}}{\omega^{i+j}-\mu^2}=0, \quad \mbox{ for $j=0,1,\cdots, n-1$}.\label{con3}
\end{eqnarray}
Notice that (\ref{con3}) is the transpose of (\ref{con2}) and indeed for the identity (\ref{con2}), if it is valid for $i=0$, then it is also valid
for all other values of $i$ since $Q^n=I_n$ and $\omega^n=1$.
So the only condition for both (\ref{con3}) and (\ref{con2}) is the second identity in  (\ref{cond}).
\end{proof}

Using the first two lemmas in the Appendix, it follows from this proposition that the Darboux matrix $M_\mu(\lambda)$ is trivial if the matrix $A$ is non-singular ($\det A\ne 0$).
\begin{Cor} If $\det A\ne 0$, then $ C^2=\mu^{2n}I_n $ and $M_\mu(\lambda)=\frac{\lambda^n-\mu^{-n}}{\lambda^n-\mu^{n}}C .$
\end{Cor}
\begin{proof} If $\det A\ne 0$, it follows from the second identity in (\ref{cond}) that 
$$
C +\mu \sum_{j=0}^{n-1}  \frac{Q^{-j} A^{\T} Q^{j}}{\omega^{j}-\mu^2}=0,
$$
which implies that matrix $A$ is diagonal since $C$ is diagonal, and thus 
$$
C =-\mu \sum_{j=0}^{n-1}  \frac{Q^{-j} A^{\T} Q^{j}}{\omega^{j}-\mu^2}=-\mu \sum_{j=0}^{n-1}  \frac{1}{\omega^{j}-\mu^2} A=\frac{n \mu^{2n-1}}{\mu^{2n}-1} A,
$$
where we used Lemma \ref{lem2} in Appendix for the last equality.
Substituting this into the first identity in (\ref{cond}), we have
$$
C^2-\frac{n}{\mu}AC=\frac{1}{\mu^{2n}}C^2=I_n, \quad \mbox{that is,} \quad C^2=\mu^{2n}I_n.
$$
Moreover, we have
$$
M_\mu(\lambda)=C+ A\sum_{i=0}^{n-1} \frac{1}{\lambda \omega^i-\mu}=C\left(1-\frac{\mu^{2n}-1}{\lambda n \mu^{2n-1}}
\sum_{i=0}^{n-1} \frac{1}{\frac{\mu}{\lambda}- \omega^i}\right)
=C\left(1-\frac{\mu^{2n}-1}{\lambda n \mu^{2n-1}}
 \frac{n (\frac{\mu}{\lambda})^{n-1} }{(\frac{\mu}{\lambda})^n- 1}\right) .
$$
Simplifying the above expression, we obtain the formula for $M_{\mu}(\lambda)$ in the statement.
\end{proof}

We define the rank of the Darboux transformation as the rank of matrix $A$.
In this paper we restrict ourself with Darboux matrices $M_{\mu}$ of rank $1$.
We represent the matrix $A$ by a bi-vector
$A=a><b$, 
where $a>=(a^{(1)},\ldots ,a^{(n)})^{\T}$ and $<b=( b^{(1)},\ldots ,b^{(n)})$ are
column and row vectors respectively and thus $A_{i,j}=a^{(i)}b^{(j)}$.

\begin{Thm}\label{cor1}
For a rank $1$ matrix $A=a><b$, and a diagonal matrix $C=\diag (p^{(i)})$, 
the matrix $M_{\mu}(\lambda)$ given by (\ref{M}) is invariant under group $\bbbd_n$ if
\begin{eqnarray}
&& {a^{(i)}}^2=\frac{({p^{(i)}}^2-1) \mu^{2(n-i)} \prod_{l=1}^{i-1} {p^{(l)}}^2 } {\mu^{2 n}-\prod_{l=1}^{n-1} {p^{(l)}}^2 };
\label{consb}\\
&&b^{(i)}=\frac{ a^{(i)} \prod_{l=i}^{n-1} {p^{(l)}}^2 ({p^{(n)}}^2-1)}{n \mu^{2(n-i)-1} p^{(i)} }, \label{coar} 
\end{eqnarray}
where $i=1, 2,  \cdots, n$ and $\prod_{l=1}^n {p^{(l)}}^2=\mu^{2 n}$.
\end{Thm}
\begin{proof} To prove this statement we need to check the conditions (\ref{cond}) in Proposition \ref{prop1} for the given matrix $A$.
It follows from Lemma \ref{lem1} that matrix $\sum_{i=0}^{n-1} Q^{-i} A^{\T} C Q^{i}$ is diagonal. 
The diagonal entries of the first identity of (\ref{cond}) are 
\begin{eqnarray}
{p^{(i)}}^2 -\frac{n}{\mu} a^{(i)} b^{(i)} p^{(i)}=1, \qquad i=1, 2, \cdots,n. \label{reab} 
\end{eqnarray}

The matrix entries of the second identity of (\ref{cond}) can be represented as
\begin{eqnarray*}
&& 0=p^{(k)} a^{(k)} b^{(l)}+\mu \sum_{j=0}^{n-1} \sum_{r=1}^n\frac{\omega^{j(r-k)} b^{(k)} {a^{(r)}}^2 b^{(l)}}{\omega^j-\mu^2}, \qquad k,l=1, 2, \cdots,n.
\end{eqnarray*}
Notice that all terms have a factor $b^{(l)}$. For a nonzero vector $b$, it is equivalent to
\begin{eqnarray}
0=p^{(k)} a^{(k)}+\mu \sum_{j=0}^{n-1} \sum_{r=1}^n\frac{\omega^{j(r-k)} b^{(k)} {a^{(r)}}^2}{\omega^j-\mu^2}, \qquad k=1, 2, \cdots,n. \label{reca}
\end{eqnarray}
Notice that there are only $n$ independent relations, from which we can determine all $b^{(k)}$ as follows:
\begin{eqnarray*}
 b^{(k)}=\frac{p^{(k)} a^{(k)} }{\mu \sum_{r=1}^n  \gamma_{r-k}(\mu^2)
{a^{(r)}}^2},
 \qquad \gamma_{r-k}(\mu^2)=\sum_{j=0}^{n-1} \frac{\omega^{j(r-k)}}
{\mu^2-\omega^j} =\frac{n \mu^{2 \{(r-k-1)\!\!\! \mod\! n\}}}{\mu^{2 n}-1},
\label{reba}
\end{eqnarray*}
which is the direct result from Lemma \ref{lem2} in Appendix.
This leads to
\begin{eqnarray}
b^{(k)}=\frac{(\mu^{2 n}-1) a^{(k)} p^{(k)}}{n \mu \sum_{r=1}^n {a^{(r)}}^2 \mu^{2 \{(r-k-1)\!\!\! \mod\! n\}}}.\label{consa}
\end{eqnarray}
These formulas for $b^{(k)}$ have appeared in \cite{bury}, where the author used them when
computing soliton solutions for the Volterra system.

Combining (\ref{reab}) and (\ref{consa}), we obtain the following linear system for ${a^{(i)}}^2$:
\begin{eqnarray*}
\mu^2 ({p^{(i)}}^2-1) \sum_{k=1,k\neq i}^n {a^{(k)}}^2 \mu^{2\{(k-i-1)\!\!\! \mod\! n\}}
+ ({p^{(i)}}^2-\mu^{2n}) {a^{(i)}}^2 =0
\end{eqnarray*}
Let ${a^{(n)}}^2=1$. We solve this linear system and obtain a unique solution
\begin{eqnarray*}
{a^{(j)}}^2=\frac{({p^{(j)}}^2-1) \mu^{2(n-j)} \prod_{l=1}^{j-1} {p^{(l)}}^2 } {\mu^{2 n}-\prod_{l=1}^{n-1} {p^{(l)}}^2 }, \qquad j=1, 2,  \cdots, n-1
\end{eqnarray*}
and the constraint $\prod_{l=1}^{n} {p^{(l)}}^2 =\mu^{2 n}$. Here we use the convention that $\prod_{l=1}^{0} {p^{(l)}}^2=1$.

Substituting (\ref{consb}) into (\ref{reab}), we get
\begin{eqnarray*}
b^{(i)}=\frac{\mu ({p^{(i)}}^2-1)}{n a^{(i)} p^{(i)}}=\frac{ a^{(i)} (\mu^{2 n}-\prod_{l=1}^{n-1} {p^{(l)}}^2 )}{n \mu^{2(n-i)-1} p^{(i)} 
\prod_{l=1}^{i-1} {p^{(l)}}^2 }
=\frac{ a^{(i)} \prod_{l=i}^{n-1} {p^{(l)}}^2 ({p^{(n)}}^2-1)}{n \mu^{2(n-i)-1} p^{(i)} },
\end{eqnarray*}
where we used $\mu^{2 n}=\prod_{l=1}^{n} {p^{(l)}}^2 $, and we complete the proof.
\end{proof}
Following from this theorem, the matrix $M_{\mu}(\lambda)$ is completely parametrised by $p^{(i)}$, $i=1, 2, \cdots, n-1$.

\section{Darboux transformation for generalised Volterra lattice}\label{Sec23}
We construct the rank $1$ Darboux matrix $M_{\mu}$, invariant under the $\bbbd_n$ reduction
group, for the Lax operator $L^{(1)}$. 

From the compatibility condition (\ref{backlundx}), we know $M_{\mu}$ satisfy
\begin{eqnarray}\label{nonzero}
D_x M_{\mu}=\cS (V^{(1)}) \ M_{\mu}-M_{\mu} \ V^{(1)}
\end{eqnarray}
First we write out the right hand of the above identity. It equals to
\begin{eqnarray*}
 &&(\lambda \bu_1 \Delta -\lambda^{-1} \Delta^{-1} \bu_1) (C+\sum_{i=0}^{n-1} \frac{Q^i A Q^{-i}}{\lambda\omega^i-\mu})
 -(C+\sum_{i=0}^{n-1} \frac{Q^i A Q^{-i}}{\lambda \omega^i-\mu})(\lambda \bu \Delta -\lambda^{-1} \Delta^{-1} \bu)\\
 &=&\lambda (\bu_1 \Delta C-C \bu \Delta )+\lambda^{-1}(C \Delta^{-1} \bu-\Delta^{-1} \bu_1 C) \\
 &&+\sum_{i=0}^{n-1}(\frac{1}{\omega^i}+\frac{\mu}{\omega^i(\lambda\omega^i-\mu)}) (\bu_1 \Delta Q^i A Q^{-i}-Q^i A Q^{-i}  \bu \Delta)\\
 &&+\sum_{i=0}^{n-1}(\frac{1}{\mu \lambda}-\frac{\omega^i}{\mu (\lambda\omega^i-\mu)})  
 (\Delta^{-1} \bu_1 Q^i A Q^{-i}-Q^i A Q^{-i} \Delta^{-1} \bu) .
\end{eqnarray*}
Notice that $\frac{1}{\omega^i} Q^{-i} \Delta Q^i=\Delta$. Then we compare the residues at different poles $\lambda=+\infty, 0, \mu$ and constant terms
on both sides of (\ref{nonzero}).
 The zero curvature condition (\ref{nonzero}) is equivalent to the following 
four identities:
\begin{eqnarray}
&&\bu_1 \Delta C-C \bu \Delta =0;\label{non1}\\
&&C \Delta^{-1} \bu-\Delta^{-1} \bu_1 C +\sum_{i=0}^{n-1}\frac{1}{\mu} 
 (\Delta^{-1} \bu_1 Q^i A Q^{-i}-Q^i A Q^{-i} \Delta^{-1} \bu)=0;\label{non3}\\
 &&A_x=\mu (\bu_1 \Delta  A - A   \bu \Delta)-\frac{1}{\mu }
 (\Delta^{-1} \bu_1  A - A  \Delta^{-1} \bu); \label{non4}\\
 &&C_x=\sum_{i=0}^{n-1}\frac{1}{\omega^i} (\bu_1 \Delta Q^i A Q^{-i}-Q^i A Q^{-i}  \bu \Delta).\label{non2}
\end{eqnarray}
It follows from (\ref{non1}) that
\begin{eqnarray}
p^{(i+1)}=p^{(i)} \exp(\phi^{(i)}-\phi_1^{(i)})=p^{(1)} \exp(\sum_{k=1}^i (\phi^{(k)}-\phi_1^{(k)})) . \label{indc} 
\end{eqnarray}
Substituting (\ref{indc}) into the identity $\prod_{i=1}^n {p^{(i)}}^2=\mu^{2n}$ from Theorem \ref{cor1}, we have
\begin{eqnarray*}
{p^{(1)}}^{2n}  \prod_{i=2}^{n} \exp(2 \sum_{k=1}^{i-1} (\phi^{(k)}-\phi_1^{(k)}))
={p^{(1)}}^{2n}   \exp(2 \sum_{k=1}^{n-1} (n-k)  (\phi^{(k)}-\phi_1^{(k)}))=
\mu^{2n},
\end{eqnarray*}
which leads to $$p^{(1)}=\mu \exp( \sum_{k=1}^{n-1} (1-\frac{k}{n})(\phi_1^{(k)}-\phi^{(k)})).
$$
Using the relation (\ref{indc}), we obtain 
\begin{eqnarray}
 p^{(i)}=\mu \exp\{(\sum_{k=i}^{n-1}-\sum_{k=1}^{n-1}\frac{k}{n}) (\phi_1^{(k)}-\phi^{(k)})\} .\label{pu}
\end{eqnarray}
It follows from Lemma \ref{lem1} in Appendix that  $$  \sum_{i=0}^{n-1} \frac{1}{\omega^{i}} Q^{i} A Q^{-i}=n \sum_{j=1}^n
a^{(j)} b^{(j-1)} {\bf e}_{j,j-1}.$$
From (\ref{non2}) we get 
\begin{eqnarray}
p^{(k)}_x=n \left( a^{(k+1)} b^{(k)} \exp(\phi_1^{(k)})-a^{(k)} b^{(k-1)} \exp( \phi^{(k-1)}) \right),\qquad k=1,2,\cdots, n-1, \label{darb} 
\end{eqnarray}
where the indexes are again counted modulo $n$. 

So far, we only deal with two of four equivalent identities of the zero curvature condition (\ref{nonzero}). We used (\ref{non1}) and (\ref{non2}) to 
obtain (\ref{pu}) and (\ref{darb}), respectively. To claim that we have obtain a Darboux transformation and further the B{\"a}cklund transformation,
we need to show that the other two identities are satisfied due to the reduction group invariance of the matrix $M_{\mu}$. 

Indeed, it can be easily checked that the identity (\ref{non3}) holds from (\ref{non1}) as follows:

We know that $\sum_{i=0}^{n-1} Q^i A Q^{-i}=n \diag\ A$ following from Lemma \ref{lem1} in Appendix. So the non-zero entries of the left side of (\ref{non3}) are
\begin{eqnarray*}
 &&p^{(i+1)} \exp(\phi^{(i)})-p^{(i)} \exp(\phi_1^{(i)})+\frac{n}{\mu} a^{(i)} b^{(i)} \exp(\phi_1^{(i)})
 -\frac{n}{\mu} a^{(i+1)} b^{(i+1)} \exp(\phi^{(i)})\\
&=&\frac{1}{ p^{(i+1)}} \exp(\phi^{(i)})-\frac{1}{ p^{(i)}}\exp(\phi_1^{(i)})=0,
\end{eqnarray*}
where we used formulas (\ref{reab}) and (\ref{indc}). 

Using (\ref{darb}), (\ref{consb}) and (\ref{coar}), we are able to check identity (\ref{non4}) is valid as shown in 
Proposition \ref{proA} in Appendix. Thus we have the following result:

\begin{Thm}\label{thm2}
The B{\"a}cklund transformation for system (\ref{2+1}) is
\begin{eqnarray}
&& (\cS\!\!-\!\!1) \phi_x^{(k)}=-\frac{n}{\mu^2} (\psh-1)^2 \left(a^{(k-1)} b^{(k)} p^{(k-1)}\exp(\phi^{(k-1)})\right),\label{flowx}
\end{eqnarray} 
where $p^{(i)}, a^{(i)}$ and $b^{(i)}$ can be expressed via (\ref{pu}), (\ref{consb}) and (\ref{coar}) in terms of $\phi^{(k)}$, 
and $\psh$ is the $n$ periodic shift operator.
\end{Thm}
\begin{proof} We have shown that the zero curvature condition (\ref{nonzero}) are equivalent to (\ref{pu}) and (\ref{darb}).
It leads from (\ref{indc}) that the dependent variables $\phi^{(k)}$, $k=1, 2,\cdots,n-1$ are determined by $p^{(i)}$ as follows:
\begin{eqnarray}
 (\cS\!\!-\!\!1) \phi^{(k)}=\ln p^{(k)}-\ln p^{(k+1)}=-(\psh-1) \ln p^{(k)}. \label{phip}
\end{eqnarray}
Differentiating both sides of (\ref{phip}) with respect to $x$ and using (\ref{darb}), we get 
\begin{eqnarray*}
 &&(\cS\!\!-\!\!1) \phi_x^{(k)}=-n (\psh-1)\left( \frac{a^{(k+1)} b^{(k)} \exp(\phi_1^{(k)})}{p^{(k)}}-\frac{a^{(k)} b^{(k-1)} \exp( \phi^{(k-1)})}{p^{(k)}} \right)\\
 &&\qquad=-n (\psh-1)\left( \frac{a^{(k+1)} b^{(k)} \exp(\phi^{(k)})}{p^{(k+1)}}-\frac{a^{(k)} b^{(k-1)} \exp( \phi^{(k-1)})}{p^{(k)}} \right)\\
 &&\qquad=-n (\psh-1)^2\left(\frac{a^{(k)} b^{(k-1)} \exp( \phi^{(k-1)})}{p^{(k)}} \right).
\end{eqnarray*}
It follows from (\ref{coar}) that $\mu^2 a^{(k)} b^{(k-1)}= a^{(k-1)} b^{(k)}  p^{(k)} p^{(k-1)}$. Substituting it into the above formula, we get
the B{\"a}cklund transformation as stated.
\end{proof}

The matrix $M_{\mu}$ satisfies the relation in Theorem \ref{cor1}, together with the relations (\ref{pu}) and (\ref{darb}) it is the Darboux transformation
for the operator $L^{(1)}$ given by (\ref{laxva}).

Obviously, the right-hand side of (\ref{flowx}) is not in the image of $\cS-1$ implying the differential-difference equation for $\phi^{(i)}$ is not evolutionary.
However, we know it is integrable due to the zero curvature condition. To directly search for its symmetries is a rather hard task for arbitrary $n$. 
In the next section, we show that there exists another Lax operator with matrix part $U_{\mu}$ such 
that the zero curvature condition between $M_{\mu}$ and $U_{\mu}$ defined by (\ref{U}) naturally leads to the symmetry flow for the 
nonevolutionary equation (\ref{flowx}).

\section{An evolutionary flow of the B{\"a}cklund transformation}\label{Sec3}
The B{\"a}cklund transformation (\ref{flowx}) given in the previous section can be viewed as an integrable differential difference equation. 
For arbitrary $n$, it is not easy to directly compute the generalised symmetries for this nonevolutionary system according to the following definition. 
In this section, we compute one of its higher symmetries using the Darboux matrix $M_{\mu}$. Meanwhile, we also provide its Lax representation.

We first give the general definition of symmetry. More details on it can be found in \cite{mr94g:58260}.
\begin{Def}
Given a $k$-component system ${\bf Q}([\bu])=0 $, we say an evolutionary vector ${\bf P}([\bu])$ is its symmetry if and only if the system
is invariant along the flow $\bu_{\tau}={\bf P}$, that is, $D_{\tau} \left({\bf Q}\right)=0$.
\end{Def}

In what follows, we construct a symmetry using the Darboux matrix $M_{\mu}$ and check the resulting flow is indeed a symmetry of (\ref{flowx}) according to the above definition.

Notice that (\ref{nonzero}) is equivalent to
\begin{eqnarray}\label{nonzerom}
 M_{\mu}^{-1} D_x (M_{\mu})=M_{\mu}^{-1} \cS (V^{(1)}) \ M_{\mu}- \ V^{(1)}.
\end{eqnarray}
Its left hand side  has simple poles at points of the generic orbit, that is, $\lambda= \omega^i \mu$ ($i=0, 1, \cdots, n-1)$, 
and  matrix $V^{(1)} $ has simple poles at points of a degenerated orbit $\lambda=0, \infty$.  We replace matrix $V^{(1)} $ by a matrix $U_{\mu}$ and require that 
$U_{\mu}$ has poles consistent to the left hand side of (\ref{nonzerom}).

In other words, we have
\begin{eqnarray}\label{lins}
 \cS \Psi=\overline{\Psi}=M_{\mu} \Psi \quad \mbox{and} \quad D_{\tau} \Psi=U_{\mu} \Psi .
\end{eqnarray}
It follows from (\ref{lins}) that matrix $U_\mu$ inherits the same reduction group (\ref{redgroup}), that is, it satisfies
\begin{eqnarray}
&& U_\mu(\lambda)=-U_\mu^{\T}(\frac{1}{\lambda}), \qquad \quad Q
U_\mu(\lambda
\omega) Q^{-1}=U_\mu(\lambda). \label{InvU}
\end{eqnarray}
Thus it is of the following form:
\begin{eqnarray}
&&U_\mu(\lambda)=\sum_{i=0}^{n-1} \left( \frac{Q^i B Q^{-i}}{\lambda
\omega^i-\mu}-\frac{Q^{-i} B^{\T} Q^{i}}{\frac{\omega^i}{\lambda} -\mu}\right),
\label{U}
\end{eqnarray}
where $B$ is an $n\times n$ matrix. 
Unlike what we have seen for the Darboux matrix  $M_\mu$,
such $U_\mu(\lambda)$ satisfying (\ref{InvU}) is  automatically invariant under the $\bbbd_n$ group.

We now derive the differential-difference equation for $p^{(i)}$, $i=1, 2, \cdots, n-1$ using the zero curvature condition 
\begin{eqnarray}\label{zero}
D_{\tau} M_{\mu}=\cS (U_{\mu}) \ M_{\mu}-M_{\mu} \ U_{\mu} .
\end{eqnarray}
Notice that its right-hand side is equal to
\begin{eqnarray*}
&&\sum_{i=0}^{n-1} \frac{Q^{-i} (B_1^{\T} C -C B^{\T}) Q^{i}}{\mu} +\sum_{i=0}^{n-1} \left( \frac{Q^i (B_1C-C B) Q^{-i}}{\lambda \omega^i-\mu} 
+\frac{Q^{-i} (B_1^{\T} C -C B^{\T}) Q^{i}}{\mu} \frac{\omega^i}{\mu \lambda-\omega^i}\right)\\
&& +\sum_{i=0}^{n-1}  \frac{Q^i (B_1 A-A B) Q^{-i}}{(\lambda \omega^i-\mu)^2}
+\sum_{i,j=0,i\neq j}^{n-1} \frac{Q^i (B_1 Q^{j-i} A -A Q^{j-i} B) Q^{-j}}{\mu (\omega^i-\omega^j)} \left(-\frac{\omega^i}{\lambda \omega^i-\mu}
+\frac{\omega^j}{\lambda \omega^j-\mu} \right)\\
&& +\sum_{i,j=0}^{n-1} \frac{Q^{i} A Q^{-(i+j)} B^{\T} Q^{j}-Q^{-j} B_1^{\T} Q^{i+j} A Q^{-i}}{\omega^{i+j}-\mu^2}\left(\frac{\mu}{\lambda \omega^i-\mu}
-\frac{\omega^j}{\lambda \mu -\omega^j}\right).
\end{eqnarray*}
Comparing to the left-hand side, we obtain its following equivalent identities:
\begin{eqnarray}
&&B_1 A-A B=0; \label{loc1}\\
&&C_{\tau}=\sum_{i=0}^{n-1} \frac{Q^{-i} (B_1^{\T} C -C B^{\T}) Q^{i}}{\mu}; \label{eqC}\\
&&\frac{ B_1^{\T} C -C B^{\T}}{\mu} = \sum_{i=0}^{n-1} \frac{Q^{i} A Q^{-i} B^{\T} -B_1^{\T} Q^{i} A Q^{-i}}{\omega^{i}-\mu^2};\label{BC}\\
&&A_{\tau}= B_1C-C B-\sum_{j=1}^{n-1} \frac{(B_1 Q^{j} A -A Q^{j} B) Q^{-j}+Q^j (B_1 Q^{-j} A-A Q^{-j} B)}{\mu (1-\omega^{j})} \nonumber\\
&&\qquad+\sum_{j=0}^{n-1}\frac{ (A Q^{-j} B^{\T} Q^{j}-Q^{-j} B_1^{\T} Q^{j} A) \mu}{\omega^{j}-\mu^2}. \label{At}
\end{eqnarray}
It is easy to check that 
\begin{eqnarray}
B=\frac{a_{-1}><b}{<b, a_{-1}>} \label{defB} 
\end{eqnarray}
is a solution of (\ref{loc1}) when $A=a><b$. 
It follows from (\ref{eqC}) and Lemma \ref{lem1} (see Appendix) that
\begin{eqnarray}
p^{(i)}_{\tau}=\frac{n}{\mu} \left(\frac{ b_1^{(i)} a^{(i)} p^{(i)}}{<b_1,\ a>}-\frac{ b^{(i)} a_{-1}^{(i)} p^{(i)}}{<b,\ a_{-1}>}\right)
=\frac{n }{\mu} p^{(i)} (\cS-1) \frac{ b^{(i)} a_{-1}^{(i)}}{<b,\ a_{-1}>}, \label{eqps}
\end{eqnarray}
where $a$ and $b$ satisfy (\ref{consb}) and (\ref{coar}).

The consistent condition (\ref{BC}) can be obtained from (\ref{defB}). To see this,  we are going to show that 
\begin{eqnarray*}
\frac{B_1^{\T} C}{\mu}  = \sum_{i=0}^{n-1} \frac{ B_1^{\T} Q^{i} A Q^{-i}}{\mu^2-\omega^{i}} \quad \mbox{and} \quad
\frac{C B^{\T}}{\mu}= \sum_{i=0}^{n-1} \frac{Q^{i} A Q^{-i} B^{\T} }{\mu^2-\omega^{i}}.
\end{eqnarray*}
Indeed, we write out the entries for the matrices by substituting (\ref{defB}) and $A=a><b$ into them.
They both are equivalent to the identity (\ref{reca}).
Finally, according to Proposition \ref{proB} in Appendix, the identity (\ref{At}) follows from (\ref{eqps}).
Thus we obtain the following result:
\begin{Thm}\label{thm3}
Let $A= a><b$ and $B$ be defined by (\ref{defB}).
The evolutionary differential-difference equation (\ref{eqps}) possesses a Lax representation (\ref{zero}) with
\begin{eqnarray*}
&&M_{\mu}=C+\sum_{i=0}^{n-1} \frac{Q^i A Q^{-i}}{\lambda \omega^i-\mu},\qquad 
U_{\mu}=\frac{Q^i B Q^{-i}}{\lambda \omega^i-\mu}-\frac{Q^{-i} B^{\T} Q^{i}}{\frac{\omega^i}{\lambda} -\mu}, \qquad 
\end{eqnarray*}
where $C=\diag (p^{(i)}), Q=\diag (\omega^i), \omega=\exp\frac{2\pi {\rm i}}{n} $, and $a$ and $b$ satisfy (\ref{consb}) and (\ref{coar}) respectively.
\end{Thm}

Using (\ref{phip}), we obtain the evolutionary differential-difference equations for $\phi^{(k)}$ from (\ref{eqps}), that is,
\begin{eqnarray}
&& \phi_{\tau}^{(i)}=-(\psh-1) \left(\frac{p_{\tau}^{(i)}}{p^{(i)}}\right)
=-\frac{n }{\mu} (\psh-1) \left(\frac{ b^{(i)} a_{-1}^{(i)}}{<b,\ a_{-1}>} \right)\label{flowt}
\end{eqnarray}
Notice that both equations (\ref{flowx}) and (\ref{flowt}) are obtained from the same Darboux matrix $M_{\mu}$. This implies that both of them share the same generalised
symmetries and conserved densities derived from the zero curvature conditions \cite{mr86h:58071, mik12, mik13}. Equation (\ref{flowt}) is evolutionary, which can be viewed
as a symmetry of the nonevolutionary equation (\ref{flowx}). We are going to show it by direct calculation in the following theorem.
\begin{Thm}\label{thm4}
Equation (\ref{flowt}) is a symmetry of the nonevolutionary equation (\ref{flowx}).
\end{Thm}
\begin{proof}
To prove this statement we need to show that equation (\ref{flowx}) is invariant along the flow (\ref{flowt}), that is, to check
$$
(\cS\!\!-\!\!1) D_x \phi_{\tau}^{(i)}=-\frac{n}{\mu^2} (\psh-1)^2 \frac{\partial }{\partial \tau} \left(a^{(i-1)} b^{(i)} p^{(i-1)}\exp(\phi^{(i-1)})\right).
$$
Therefore, we only need to show
\begin{eqnarray}
\mu (\cS\!\!-\!\!1) D_x \left(\frac{ b^{(i)} a_{-1}^{(i)}}{<b,\ a_{-1}>} \right)
= (\psh-1) \frac{\partial }{\partial \tau} \left(a^{(i-1)} b^{(i)} p^{(i-1)}\exp(\phi^{(i-1)})\right).\label{prom}
\end{eqnarray}
Using the expressions for $a_x^{(i)}$ and $b_x^{(j)}$ in the proof of Proposition \ref{proA} in Appendix, we get
\begin{eqnarray}
&&b_x^{(i)} a_{-1}^{(i)}+b^{(i)} a_{-1,x}^{(i)}=
(\psh-1)\left((\frac{a_{-1}^{(i-1)} b^{(i)} }{\mu}+\mu a_{-1}^{(i)} b^{(i-1)}) \exp(\phi^{(i-1)})\right)\nonumber\\
&&\qquad + \mu a_{-1}^{(i)} b^{(i)} (\cS-1) \left(\frac{ a_{-1}^{(1)} p_{-1}^{(n)} \exp(\phi_{-1}^{(n)})}{p_{-1}^{(1)} a_{-1}^{(n)} }
-\frac{ b_{-1}^{(n-1)} \exp(\phi_{-1}^{(n-1)}) } {{p_{-1}^{(n)}}^2 b_{-1}^{(n)}} \right) .\label{abx}
\end{eqnarray}
This leads to
\begin{eqnarray*}
&& \frac{\partial}{\partial x} \frac{b^{(i)} a_{-1}^{(i)}}{<b,\ a_{-1}>}
=\frac{b_x^{(i)} a_{-1}^{(i)}+b^{(i)} a_{-1,x}^{(i)}}{<b,\ a_{-1}>}-\frac{b^{(i)} a_{-1}^{(i)}(<b,\ a_{-1,x}>+<b_x,\ a_{-1}>)}{<b,\ a_{-1}>^2}\\
&&\qquad= \frac{1}{<b,\ a_{-1}>}(\psh-1)\left((\frac{a_{-1}^{(i-1)} b^{(i)} }{\mu}+\mu a_{-1}^{(i)} b^{(i-1)}) \exp(\phi^{(i-1)})\right) .
\end{eqnarray*}
To prove the identity (\ref{prom}), we now only need to show that
\begin{eqnarray*}
(\cS\!\!-\!\!1) \frac{1}{<b,\ a_{-1}>} (a_{-1}^{(i-1)} b^{(i)} +\mu^2 a_{-1}^{(i)} b^{(i-1)}) \exp(\phi^{(i-1)})
=  \frac{\partial }{\partial \tau} \left(a^{(i-1)} b^{(i)} p^{(i-1)}\exp(\phi^{(i-1)})\right), 
\end{eqnarray*}
which is the identity (\ref{comm1}) proved in Appendix. 
\end{proof}

\section{Formal diagonalisation of the Lax-Darboux scheme and Conservation laws}\label{Sec4}
In this section, we study the conservation laws for integrable equation (\ref{flowx}). Given the Lax representation of an equation, its
conservation laws, both conserved densities and conserved fluxes, can be computed by preforming formal diagonalisation of its Lax pair \cite{mr86h:58071}.
This idea has been adapted to differential difference and partial difference equations. See, for example, \cite{mik12, mik13}.
Instead of directly working on the Lax pair of (\ref{flowx}), we work on the Lax pair of its symmetry. 
Then we prove that the obtained conserved densities are the conserved densities of (\ref{flowx}).

Consider the gauge transformation $\Psi=W(\lambda) {\tilde \Psi}$. It follows from (\ref{lins}) that the Lax operators transform into 
\begin{eqnarray*}
 M_\mu(\lambda)&\mapsto &
\cM_\mu(\lambda)=\cS(W^{-1}(\lambda))M_\mu(\lambda) W(\lambda);\\
 U_\mu(\lambda)&\mapsto &
 \cU_\mu(\lambda)=W^{-1}(\lambda)U_\mu(\lambda)W(\lambda)-W^{-1}(\lambda)
D_\tau(W(\lambda)).
\end{eqnarray*}
In this section we shall show that both operator $U_\mu(\lambda)$  and the Darboux matrix
$M_\mu(\lambda)$ defined in Theorem \ref{thm3} can be simultaneously brought to a
formal block-diagonal form by a suitable transformation $W(\lambda)$. In other words, we show that there exists $W(\lambda)$
such that
\begin{eqnarray*}
 \cU_\mu(\lambda)&=&\frac{\bpi}{ \lambda-\mu }
+\cU_0+(\lambda-\mu)\cU_1+(\lambda-\mu)^2\cU_2+\cdots;\\
\cM_\mu(\lambda)&=&\frac{<b,\ a_{-1}>\bpi}{\lambda-\mu}
+\cM_0+(\lambda-\mu)\cM_1+(\lambda-\mu)^2\cM_2+\cdots,
\end{eqnarray*}
where $\bpi$ is a matrix (a projector) with 1 at the $(1,1)$ position and
$0$ elsewhere, and also  $\ad_\bpi \cU_k=\ad_\bpi \cM_k=0,\ k=0,1,2,\ldots $,
i.e. the coefficient $n\times n$ matrices $\cU_k$ and $\cM_k$ have a
block-diagonal form
\[ \left(\begin{array}{cccc}
          *&0&\cdots &0\\
0&*&\cdots&*\\
\vdots&\vdots& &\vdots\\
0&*&\cdots&*
         \end{array}\right).
\]
We shall also show that the entries of $\cM_k$ and $\cU_k$ are local, which means that they can be expressed in terms of variables $\phi^{(i)}$ and their $\cS$-shifts.

Obviously, it follows from (\ref{zero}) that the transformed operators also satisfy the zero curvature condition
\begin{equation}
 D_{\tau} \cM_\mu(\lambda)=\cS(\cU_\mu(\lambda))\cM_\mu(\lambda)-\cM_\mu(\lambda)
\cU_\mu(\lambda) . 
\label{McM}
\end{equation}
The projection to the element $(1,1)$ leads to 
\[
  D_\tau m(z)=m(z)(\cS-1)u(z).
\]
Here we introduced a new parameter $z=\lambda-\mu$ and denoted
$$m(z)=(\cM_\mu(z+\mu))_{1,1} \quad  \mbox{and} \quad  u(z)=(\cU_\mu(z+\mu))_{1,1}.$$ 
Thus $\log m(z)$ and $u(z)$ are generating
functions of local conservation laws and corresponding fluxes
\begin{equation}\label{consrho}
D_\tau \rho_k=(\cS-1)\sigma_k,\qquad \log m(z)=-\log(z)+\sum_{k=0}^{\infty}
z^k \rho_k,\quad  u(z)=z^{-1}+\sum_{k=0}^\infty z^k \sigma_k.
\end{equation}
We represent the transformation $W(z+\mu)$ in the form 
\[
 W=W_*\cW,\qquad \cW=I+z W_1+z^2 W_2+z^3 W_3+\cdots,
\]
 where $W_*$ is a
$z$-independent invertible matrix 
and $\cW$ is a formal series in $z$ with ``off block-diagonal'' coefficients
$W_k$ (in the image of $ \ad_\bpi$).  The entries of $W_*$ and $W_k$ are local. 
The gauge transformation $W_*$ brings simultaneously the
residues 
\begin{equation}\label{resUM}
 \res_{\lambda=\mu} U_\mu
(\lambda)=\frac{a_{-1}><b}{<b,\, a_{-1}>}, \qquad
\res_{\lambda=\mu}M_\mu(\lambda)=a><b
\end{equation}
to the diagonal form
$$  W_*^{-1}\frac{a_{-1}><b}{<b,\, a_{-1}>}W_* =\bpi,\quad
\cS(W_*^{-1})a><bW =<b,\ a_{-1}>\bpi,$$
while the formal series $\cW$ takes care of the regular in $z$ parts.

To construct the gauge transformation $W_*$ we note that vector $<b$ is non-zero
and therefore for some $k$ it has a non-zero component $b^{(k)}\ne 0$. Let
$\alpha_k>$ denotes a vector-column with 1 at the $k$-th position and zeros
elsewhere and $<\alpha_k$ is the transpose of $\alpha_k>$, i.e., $<\alpha_k=\left(\alpha_k>\right)^{\T}$. 
Then matrix $W_*$ can be written in the
form
\[
 W_*=\left(b^{(k)} I-\alpha_k><b+a_{-1}><\alpha_k\right)\Delta^{1-k}
\]
where $\Delta$ is a  matrix defined in (\ref{Delta}). One can check that $\det
W_*=(b^{(k)})^{n-2}<b\, ,\,a_{-1}>\ne 0$

The gauge transformation $W_*$
brings the singular parts of $U_\mu(z+\mu)$ and $M_\mu(z+\mu)$ to a diagonal
form
\begin{eqnarray*}
 \hat{\cU}(z)&=&W_*^{-1}U_\mu(z+\mu)W_*-W_*^{-1}D_\tau(W_*)=z^{-1}\bpi
+\hat{\cU}_0+z\hat{\cU}_1+z^2\hat{\cU}_2+\cdots\\
\hat{\cM}(z)&=&\cS(W_*^{-1})M_\mu(z+\mu) W_*=z^{-1}<p_{-1}q>\bpi
+\hat{\cM}_0+z\hat{\cM}_1+z^2\hat{\cM}_2+\cdots\ .
\end{eqnarray*} 
The coefficients of the regular part can be easily found, but they are not yet
in the block-diagonal form. For example, we obtain
\begin{eqnarray*}
 &&\hat{\cU}_0=\frac{1}{<b\, ,\,
a_{-1}>}W_*^{-1}\left(\sum_{i=1}^{n-1}\frac{ Q^ia_ { -1 } ><bQ^ { -i } } { \mu
(\omega^i-1)}-\sum_{i=1}^{n}\frac{\mu
Q^{-i}b><a_{-1}Q^i}{\omega^i-\mu^2}\right)W_*-W_*^{-1}D_\tau(W_{*}),\\
&&\hat{\cM}_0=\cS(W_*^{-1})\left(C+\sum_{i=1}^{n-1}\frac{ Q^ia ><bQ^ { -i } } {
\mu
(\omega^i-1)}\right)W_*\, .
\end{eqnarray*}
In particular, we have
\begin{eqnarray}
 &&(\hat{\cU}_0)_{1,1}=\! - \frac{<b,\, a_{-1.\tau}>}{<b,\,
a_{-1}>}\!  -\! \frac{1}{\mu <b,\,
a_{-1}>^2}\! \sum_{k,q=1}^n\! \! 
\left(\!\Gamma_{k-q}b^{(k)}a_{-1}^{(k)}b^{(q)}a_{-1}^{(q)}-\mu^2\gamma_{q-k}
(\mu^2)(b^ { (k)}a_{-1}^{(q)})^2\right)\!,\label{uhat11}\\
&&(\hat{\cM}_0)_{1,1}=\frac{<b_1,\, C a_{-1}>}{<b_1\, ,\,
a>}  - \frac{1}{\mu <b_1\, ,\,
a>}\sum_{k,q=1}^n \Gamma_{k-q}b_1^{(k)}a^{(k)}b^{(q)}a_{-1}^{(q)},
\label{mhat11}
\end{eqnarray}
where the functions $\Gamma_k,\ \gamma_k(x)$ were introduced in Lemma \ref{lem2} in Appendix.

Now we can construct a formal series $\cW$ which transforms simultaneously
$\hat{\cU}(z)$ and $\hat{\cM}(z)$ into a block-diagonal form.
Actually, it follows from equation (\ref{McM}) that if we formally (block)
diagonalise the Darboux matrix, then the corresponding Lax operator must also be
block-diagonal and vice-versa.

\begin{Pro}\label{pro2}
 There exists a unique formal series $\cW=I+z W_1+z^2 W_2+z^3 W_3+\cdots$ with
$W_k\in \ad_\bpi \mbox{Mat}_{n\times n}(\cF)$ such that all coefficients
$\cM_k$ of
\begin{equation}
 \label{procW}
\cM(z)=\cS(\cW^{-1})\hat{\cM}(z) \cW=z^{-1}<b, \ a_{-1}>\bpi
+\cM_0+z\cM_1+z^2\cM_2+\cdots\ 
\end{equation}
are block-diagonal.
\end{Pro} 
\begin{proof}
All matrices in the equation (\ref{procW}) we split into four blocks associated
with the projector$\bpi$:
\[
 \cM(z)=\left(\begin{array}{cc}
m&<0 \\             ×
0>&N            \end{array}\right),\ 
\hat{\cM}(z)=\left(\begin{array}{cc}
\hat{m}&<f \\             ×
g>&\hat{N}             \end{array}\right),\
\cW =\left(\begin{array}{cc}
1&<q \\             ×
r >&I_{n-1}             \end{array}\right),\
\]
where $0>,g >,r >$  and $<0,<f ,<q $ are $n-1$ dimensional column and row
vectors respectively, $N, \hat{N} , I_{n-1}$ are square $(n-1)\times (n-1)$ matrices.

We rewrite equation (\ref{procW})  in the form 
\begin{equation}\label{procW1}
 \hat{\cM}(z) \cW=\cS(\cW)\cM(z)
\end{equation}
and split it into four blocks as above. It leads to one scalar, two vector and
one matrix equations 
\begin{eqnarray}\label{11}
 &&\hat{m}+<f\, ,\, r>=m,\\ \label{21}
&& \hat{N}\, r>+g>=\cS(r>)m,\\ \label{12}
&& \hat{m} <q+<f=\cS(<q)N,\\ \label{22}
&&g><f+\hat{N}=N.
\end{eqnarray}
Using (\ref{11}) and (\ref{22}) we eliminate $m$ and $N$ from ((\ref{21}) and
(\ref{12})
\begin{eqnarray}\label{req}
 \hat{N}\, r>+g>=\cS(r>)\hat{m}+\cS(r>)<f\, ,\, r>\\
 \hat{m} <q+<f=\cS(<q)g><f+\cS(<q)\hat{N}
\end{eqnarray}
Substitution of the formal  expansions
\begin{eqnarray*}
\hat{m}=z^{-1}<b,\ a_{-1}>+\hat{m}_0+z\hat{m}_1+z^2\hat{m}_2+z^3\hat{m}_3\cdots\
,\\
\hat{N}=\hat{N}_0+z\hat{N}_1+z^2\hat{N}_2+z^3\hat{N}_3\cdots\ ,\\
r>=zr_1>+z^2r_2>+z^3r_3>\cdots\ ,\\
<q=z<q_1+z^2<q_2+z^3<q_3\cdots\
\end{eqnarray*}
leads to recurrence relations for determining elements $r_k>$ and
$<q_k$. From (\ref{req}) it follows that 
\begin{equation}
 r_{k+1}>=\frac{1}{<b_{-1}\, ,\,
a_{-2}>}\cS^{-1}\left(g_k>+\sum_{i=0}^{k-1}(\hat{N}_i
r_{k-i}>-\cS(r_{k-i}>\hat{m}_i)-\sum_{i,j=1}^{i+j\le k}\cS(r_i>)<f_{k-i-j}\,
,\, r_j>\right).
\end{equation}
Then the coefficients $m_k$ in the expansion 
\[
 m=z^{-1}<b,\ a_{-1}>+{m}_0+z{m}_1+z^2{m}_2+z^3{m}_3\cdots\
\]
are determined by (\ref{11})
\begin{equation}
\label{mexp}
 m_k=\hat{m}_k+\sum_{i=1}^k <f_{k-i}\, ,\, r_i>. 
\end{equation}
 Coefficients  ${\bf N}_k$ in the expansion ${\bf N}={\bf N}_0+z{\bf
N}_1+z^2{\bf N}_2+\cdots$ and $<q_k$ can be uniquely found from (\ref{12}),
(\ref{22}) in a similar way.
\end{proof}
It follows from (\ref{mexp}) and (\ref{mhat11}) that 
\[
 m=z^{-1}<b,\ a_{-1}> - \frac{1}{\mu <b_1\, ,\,
a>}\sum_{k,q=1}^n \Gamma_{k-q}b_1^{(k)}a^{(k)}b^{(q)}a_{-1}^{(q)}
+\cO(z)
\]
and thus the conserved densities densities (\ref{consrho}) for equation (\ref{eqps}) are of the form 
\begin{eqnarray}
&& \rho_0=\log <b,\ a_{-1}> ,\label{cons1} \\
&& \rho_1=-\frac{1}{\mu <b_1 ,\,a><b,\, a_{-1}> }\sum_{k,q=1}^n
\Gamma_{k-q}b_1^{(k)}a^{(k)}b^{(q)}a_{-1}^{(q)} .\label{cons2}
\end{eqnarray}

We have formally diagonalised the Darboux matrix $M_{\mu}$ in Proposition \ref{pro2}. It follows from (\ref{nonzero}) that
the corresponding Lax operator $V^{(1)}$ must also be block-diagonal. Therefore, we have the following result:
\begin{Cor} The above expressions (\ref{cons1}) and (\ref{cons2}) are conserved densities for the integrable nonevolutionary equation (\ref{flowx})
given in Theorem \ref{thm2}, where $a$ and $b$ are defined in Theorem \ref{cor1}.
\end{Cor}
In fact, we can directly check $D_x \rho_k \in {\rm Im}(\cS-1)$ for $k=0,1$ for equation (\ref{flowx}). Here we only present it for $k=0$. Using (\ref{abx}) we have
\begin{eqnarray*}
&&D_x \rho_0=\frac{1}{<b,\ a_{-1}>} \sum_{i=1}^n\left(b_x^{(i)} a_{-1}^{(i)}+b^{(i)} a_{-1,x}^{(i)}\right)\\
&&\quad=\mu (\cS-1) \left(\frac{ a_{-1}^{(1)} p_{-1}^{(n)} \exp(\phi_{-1}^{(n)})}{p_{-1}^{(1)} a_{-1}^{(n)} }
-\frac{ b_{-1}^{(n-1)} \exp(\phi_{-1}^{(n-1)}) } {{p_{-1}^{(n)}}^2 b_{-1}^{(n)}} \right) .
\end{eqnarray*}

\section{Discussion}\label{Sec5}
In this paper, we construct a Darboux transformation with Dihedral reduction group for the $2$-dimensional generalisation of the Volterra
lattice (\ref{2+1}) with period  $n$.  The reduction group enables us to parametrise the Darboux matrix by $n-1$ dependent variables.
The Dihedral reduction group is generated by both inner automorphisms and outer automorphisms. 
To the best of our knowledge, this is the first example
to deal with arbitrary $n$ and the outer automorphisms. 

The B{\"a}cklund transformation resulting from the Darboux transformation can be viewed as a nonevolutionary multi-component integrable differential difference
equation. Assuming the Lax operator having the same simple poles as the logarithmic derivative of Darboux matrix, we obtain its local higher symmetry. 
In the similar way, local symmetries can be found for the ABS equations \cite{abs,X}. 

In this paper, we only investigated the conservation laws of this nonevolutionary multi-component integrable differential difference
equation. To obtain them we transform the Darboux matrix and the Lax operator into block-diagonal form.
The corresponding recursion operator and bi-Hamiltonian structure are not studied yet.
Knowing the Lax representation, in principle, we are able to construct the recursion operator for  a fixed period $n$ as it was done for equation 
(\ref{2+1}) in \cite{wang09} applying the method given in \cite{mr2001h:37146}. It would be interesting to see whether that can be done for arbitrary
$n$ as in the case for the Narita-Itoh-Bogoyavlensky lattice \cite{wang12}.

Following from the Bianchi commutativity, we can construct a new integrable discrete equation with Dihedral reduction group using the Darboux matrix, which 
we have not included in this paper since we could not write it down in a neat way. The integrable nonevolutionary system and the symmetry flow are its nonlocal and 
local symmetry, respectively. It would be interesting to find continuous limits of the systems obtained.

\section*{Appendix: Technical results used in the proofs}
We first give two simple lemmas, which will be used to simplify expression throughout the paper.
We then give the detailed computation to check the consistency required for the zero curvature conditions.
\begin{Lem}\label{lem1} Let $P$ be an $n\times n$ matrix and $\omega=\exp\frac{2\pi i}{n}$. Then 
\begin{equation}
 \frac{1}{n} \sum_{m=1}^{n} \omega^{km} Q^{m} P Q^{-m}=\sum_{i=1}^{n}
P_{i,i+k}{\bf e}_{i,i+k},
\end{equation}
where all indexes are counted modulo $n$, matrix ${\bf e}_{i,j}$ has a unit
entry at the position  $(i,j)$ and zero elsewhere.  In the case $k=0$ it is a
projection to the diagonal part of the matrix $P$, i.e.,
$$ \frac{1}{n} \sum_{m=1}^{n}
Q^{m} P Q^{-m}=\mbox{diag}\, P.$$
\end{Lem}
\begin{proof}
Let us compute the $(i,j)$ entry of $\sum_{m=1}^{n}\omega^{km} Q^{m} P 
Q^{-m}$, which equals to
\begin{eqnarray*}
\sum_{m=1}^{n} \omega^{ km+im-jm} P_{i,j}
=\left\{\begin{array}{ll} 0 & k+i-j\not\equiv 0\ \mbox{mod}\, n \\
     n P_{i,j}   & k+i-j\equiv 0\ \mbox{mod}\, n 
       \end{array}\right.\, .
\end{eqnarray*}
This is exactly the right hand of the identity in the statement.
\end{proof}

\begin{Lem}\label{lem2}
Let $\omega=\exp\frac{2\pi i}{n}$. Then 
$$\gamma_l(x)=\sum_{j=0}^{n-1} \frac{\omega^{lj}}{x-\omega^j}=\frac{n
x^{(l-1)\!\!\! \mod\! n}}{x^n-1}
\quad \mbox{and} \quad \Gamma_l=\sum_{j=1}^{n-1}
\frac{\omega^{lj}}{1-\omega^j}=(l-1)\!\!\! \mod\! n-\frac{n-1}{2}.$$
\end{Lem}
\begin{proof}
We first prove the identity for $\gamma_0(x)$. Since $\omega$ is a primitive
$n^{\rm th}$ root of unity, we have 
\begin{eqnarray}\label{prod}
 x^n-1=\prod_{j=0}^{n-1}(x-\omega^j)
\end{eqnarray}
We now take logarithm of both sides of (\ref{prod}) and get
$
\ln (x^n-1)=\sum_{j=0}^{n-1} \ln (x-\omega^j).
$
We then differentiate it with respect to $x$. This leads to the value for $l=0$, that is,
\begin{eqnarray*}
 \gamma_0(x)=\sum_{j=0}^{n-1} \frac{1}{x-\omega^j}=\frac{
nx^{n-1}}{x^n-1}=\frac{n x^{(-1)\!\!\! \mod\! n}}{x^n-1}.
\end{eqnarray*}
Notice that  $(\omega^{lj}-x^l)=(\omega^j-x) \sum_{r=0}^{l-1} \omega^{rj} x^{l-1-r}$ for $0<l<n-1$. Thus we have
\begin{eqnarray*}
&& \gamma_l(x)=\sum_{j=0}^{n-1} \frac{\omega^{lj}}{x-\omega^j}=\sum_{j=0}^{n-1}
\frac{x^l}{x-\omega^j}+\sum_{j=0}^{n-1} \frac{\omega^{lj}-x^l}{x-\omega^j}\\
&&\quad = x^l \gamma_0(x)-\sum_{j=0}^{n-1} \sum_{r=0}^{l-1} \omega^{rj}
x^{l-1-r}=\frac{n x^{n+l-1}}{x^n-1}-n x^{l-1}=\frac{n x^{l-1}}{x^n-1}
=\frac{n x^{(l-1)\!\!\! \mod\! n}}{x^n-1}.
\end{eqnarray*}
We know that  $\omega^r=\omega^l$ if $l-r\equiv0 \!\! \mod \! n$. Thus
$\gamma_l(x)=\gamma_r(x)$ if $r\equiv l\!\!\mod \! n$.

We now prove the second part of the statement based on the formula for
$\gamma_l(x)$. We have
\begin{eqnarray*}
\Gamma_l=\sum_{j=1}^{n-1} \frac{\omega^{lj}}{1-\omega^j}&=&\lim_{x\rightarrow
1}\sum_{j=1}^{n-1} \frac{\omega^{lj}}{x-\omega^j}=
\lim_{x\rightarrow 1}\left(\frac{n x^{(l-1)\!\!\! \mod\! n}}{x^n-1}-\frac{1}{x-1} \right)\\
&=&\lim_{x\rightarrow 1}\frac{n x^{(l-1)\!\!\! \mod\! n}-\sum_{l=0}^{n-1}x^l}{x^n-1}
=(l-1)\!\!\! \mod\! n-\frac{n-1}{2}
\end{eqnarray*}
and hence we complete the proof.
\end{proof}

In Theorem \ref{cor1}, we give the relation between the entries of matrix $A$ and those of matrix $C$. In the following lemma we give the formulas 
on their derivatives.
\begin{Lem}\label{lem3}
Consider $p^{(j)}$, $j=1,2, \cdot, n$, are the smooth function of variable $\tau$. Then
\begin{eqnarray}
&&a_{\tau}=a^{(i)} \left(\frac{p^{(i)} p_{\tau}^{(i)} } { {p^{(i)}}^2-1 }
+ \sum_{r=1}^{i-1} \frac{p_{\tau}^{(r)}}{p^{(r)}} -\frac{ p_{\tau}^{(n)} } 
{p^{(n)}( {p^{(n)}}^2 -1) }\right); \label{atau}\\
&&b_{\tau}=b^{(i)}\left(\frac{a_{\tau}^{(i)}}{ a^{(i)}}+\sum_{r=i}^{n-1}\frac{2 p_{\tau}^{(r)}}{p^{(r)}}
-\frac{p_{\tau}^{(i)}}{p^{(i)}} +\frac{2 p^{(n)} p_{\tau}^{(n)} } 
{ {p^{(n)}}^2 -1 }\right). \label{btau}
\end{eqnarray}
\end{Lem}
\begin{proof} We differentiate both sides of formula (\ref{consb}) in Theorem \ref{cor1} with respect to $\tau$ and we get
\begin{eqnarray*}
&&a_{\tau}^{(i)}=\frac{p^{(i)} p_{\tau}^{(i)} \mu^{2(n-i)} \prod_{l=1}^{i-1} {p^{(l)}}^2 } { a^{(i)}(\mu^{2 n}-\prod_{l=1}^{n-1} {p^{(l)}}^2) }
+\frac{({p^{(i)}}^2-1) \mu^{2(n-i)} \prod_{l=1}^{i-1} {p^{(l)}}^2 \sum_{r=1}^{i-1} \frac{p_{\tau}^{(r)}}{p^{(r)}}} { a^{(i)}(\mu^{2 n}-\prod_{l=1}^{n-1} {p^{(l)}}^2) } \nonumber\\
&&\qquad-\frac{({p^{(i)}}^2-1) \mu^{2(n-i)} \prod_{l=1}^{i-1} {p^{(l)}}^2 p_{\tau}^{(n)} \prod_{r=1}^{n-1} {p^{(r)}}^2} 
{a^{(i)} p^{(n)}(\mu^{2 n}-\prod_{l=1}^{n-1} {p^{(l)}}^2)^2 }.
\end{eqnarray*}
Here we used the identity $\{\prod_{r=1}^{n-1} {p^{(r)}}^2\}_{\tau}=-\frac{2 p_{\tau}^{(n)} }{p^{(n)}} \prod_{r=1}^{n-1} {p^{(r)}}^2$ obtained directly from 
differentiating $\prod_{r=1}^{n} {p^{(r)}}^2=\mu^{2n} $ with respect to $\tau$. We then substitute formula (\ref{consb}) into it to get (\ref{atau}).

Similarly, we differentiate both sides of formula (\ref{coar}) in Theorem \ref{cor1} with respect to $\tau$ and we have
\begin{eqnarray*}
&&b_{\tau}^{(i)}=\frac{ a_{\tau}^{(i)} \prod_{l=i}^{n-1} {p^{(l)}}^2 ({p^{(n)}}^2-1)}{n \mu^{2(n-i)-1} p^{(i)} }
+\frac{a^{(i)} \prod_{l=i}^{n-1} {p^{(l)}}^2 ({p^{(n)}}^2-1)}{n \mu^{2(n-i)-1} p^{(i)} } \left(\sum_{r=i}^{n-1}\frac{2 p_{\tau}^{(r)}}{p^{(r)}}
-\frac{p_{\tau}^{(i)}}{p^{(i)}}\right)\\
&&\qquad+\frac{2 a^{(i)} \prod_{l=i}^{n-1} {p^{(l)}}^2 p^{(n)} p_{\tau}^{(n)}}{n \mu^{2(n-i)-1} p^{(i)} }.
\end{eqnarray*}
Substitute formula (\ref{coar}) into it we obtain (\ref{btau}).
\end{proof}

\begin{Pro}\label{proA}
For the matrix $M_{\mu}$ satisfying Theorem \ref{cor1}, the identity (\ref{non4}) holds given formula (\ref{darb}).
\end{Pro}
\begin{proof}
To compute the left-hand side of (\ref{non4}), we first compute $a_x$ and $b_x$. Using the above lemma, Theorem \ref{cor1} and (\ref{darb}), we have
\begin{eqnarray*}
&&a_x^{(i)} =a^{(i)} \left(\frac{p^{(i)} p_x^{(i)} } { {p^{(i)}}^2-1 }
+ \sum_{r=1}^{i-1} \frac{p_x^{(r)}}{p^{(r)}} -\frac{ p_x^{(n)} } 
{p^{(n)}( {p^{(n)}}^2 -1) }\right)\\
&&\quad = n a^{(i)} \frac{p^{(i)} \left( a^{(i+1)} b^{(i)} \exp(\phi_1^{(i)})-a^{(i)} b^{(i-1)} \exp(\phi^{(i-1)}) \right) } { {p^{(i)}}^2-1 }
\quad ({p^{(i)}}^2-1=\frac{n}{\mu}a^{(i)}b^{(i)}p^{(i)}) \\
&&\qquad+n a^{(i)} \sum_{r=1}^{i-1} \frac{ a^{(r+1)} b^{(r)} \exp(\phi_1^{(r)})-a^{(r)} b^{(r-1)} \exp(\phi^{(r-1)}) }{p^{(r)}} 
\quad (\frac{\exp(\phi_1^{(r)})}{p^{(r)}}=\frac{\exp(\phi^{(r)})}{p^{(r+1)}})\\
&&\qquad-n a^{(i)}\frac{ \left( a^{(1)} b^{(n)} \exp(\phi_1^{(n)})-a^{(n)} b^{(n-1)} \exp(\phi^{(n-1)}) \right) } 
{p^{(n)}( {p^{(n)}}^2 -1) } \quad\quad ({p^{(n)}}^2-1=\frac{n}{\mu}a^{(n)}b^{(n)}p^{(n)})\\
&&\quad = \mu a^{(i+1)}  \exp(\phi_1^{(i)})-  \frac{\mu a^{(i)} b^{(i-1)} \exp(\phi^{(i-1)}) } { b^{(i)} }
 + \frac{n {a^{(i)}}^2 b^{(i-1)} \exp(\phi^{(i-1)})}{p^{(i)}}
\\
&&\qquad\ \ -\frac{n {a^{(i)}} a^{(1)} b^{(n)} \exp(\phi^{(n)})}{p^{(1)}} -\frac{\mu a^{(i)} a^{(1)} \exp(\phi^{(n)})}{p^{(1)} a^{(n)} p^{(n)}} 
+\frac{\mu a^{(i)} b^{(n-1)} \exp(\phi^{(n-1)}) } 
{{p^{(n)}}^2 b^{(n)}} \\
&&\quad = \mu a^{(i+1)}  \exp(\phi_1^{(i)}) \!- \! \frac{\mu a^{(i)} b^{(i-1)} \exp(\phi^{(i-1)}) } { b^{(i)} {p^{(i)}}^2 } 
\!-\! \frac{\mu a^{(i)} a^{(1)} p^{(n)} \exp(\phi^{(n)})}{p^{(1)} a^{(n)} } 
+\frac{\mu a^{(i)} b^{(n-1)} \exp(\phi^{(n-1)}) } 
{{p^{(n)}}^2 b^{(n)}}
\end{eqnarray*}
and
\begin{eqnarray*}
&&b_x^{(j)}=b^{(j)}\left(\frac{a_x^{(j)}}{ a^{(j)}}+\sum_{r=j}^{n-1}\frac{2 p_x^{(r)}}{p^{(r)}}
-\frac{p_x^{(j)}}{p^{(j)}} +\frac{2 p^{(n)} p_x^{(n)} } { {p^{(n)}}^2 -1 }\right)\\
&&\quad =\frac{ \mu a^{(j+1)} b^{(j)} \exp(\phi_1^{(j)}) }{ a^{(j)}}
 \!- \! \frac{\mu  b^{(j-1)} \exp(\phi^{(j-1)}) } {  {p^{(j)}}^2 } 
\!-\! \frac{\mu b^{(j)} a^{(1)} p^{(n)} \exp(\phi^{(n)})}{p^{(1)} a^{(n)} } 
\\
&&\qquad\ \ 
+\frac{\mu b^{(j)} b^{(n-1)} \exp(\phi^{(n-1)}) } 
{{p^{(n)}}^2 b^{(n)}} 
 +\frac{2 n {b^{(j)}} a^{(n)} b^{(n-1)} \exp(\phi^{(n-1)})}{p^{(n)}} 
-\frac{ n {b^{(j)}} a^{(j)} b^{(j-1)} \exp(\phi^{(j-1)})}{p^{(j)}}
\\&&\qquad\ \
-\frac{ n {b^{(j)}}^2 a^{(j+1)}  \exp(\phi_1^{(j)})}{p^{(j)}}
+\frac{2 \mu b^{(j)} a^{(1)} \exp(\phi_1^{(n)})}{ a^{(n)}} 
-\frac{2 \mu b^{(j)} b^{(n-1)} \exp(\phi^{(n-1)}) } 
{ b^{(n)}}\\
&&\quad =\frac{ \mu a^{(j+1)} b^{(j)} \exp(\phi_1^{(j)}) }{ a^{(j)} {p^{(j)}}^2}
 \!- \! \mu  b^{(j-1)} \exp(\phi^{(j-1)}) 
\!+\! \frac{\mu b^{(j)} a^{(1)} p^{(n)} \exp(\phi^{(n)})}{p^{(1)} a^{(n)} } 
-\frac{\mu b^{(j)} b^{(n-1)} \exp(\phi^{(n-1)}) } 
{{p^{(n)}}^2 b^{(n)}} .
\end{eqnarray*}
Notice that $\mu^2 a^{(i+1)} b^{(i)}= a^{(i)} b^{(i+1)}  p^{(i)} p^{(i+1)}$ following from (\ref{coar}). Therefore, we have 
\begin{eqnarray*}
&&\frac{\partial (a^{(i)} b^{(j)})}{\partial x}=a_x^{(i)} 
b^{(j)} + a^{(i)} b_x^{(j)}\\
&=&\mu \left(\exp(\phi_1^{(i)}) a^{(i+1)}b^{(j)}-a^{(i)}b^{(j-1)} \exp(\phi^{(j-1)})\right)
-\frac{1}{\mu } \left(\exp(\phi_1^{(i-1)}) a^{(i-1)}b^{(j)}-a^{(i)}b^{(j+1)} \exp(\phi^{(j)})\right),
\end{eqnarray*}
which is the $(i,j)$ entry of the right-hand side of (\ref{non4}) and thus we prove the statement.
\end{proof}
\begin{Pro}\label{proB}
For the matrix $M_{\mu}$ satisfying Theorem \ref{cor1}, the identity (\ref{At}) holds if $p^{(i)}$ satisfy (\ref{eqps}).
\end{Pro}
\begin{proof} We write out the $(l,k)$ entry of the right-hand side of (\ref{At}), simplify it using Lemma \ref{lem2} and (\ref{reca}) and obtain
\begin{eqnarray}
&&\frac{ a^{(l)} b_1^{(k)} p^{(k)} }{<b_1,\ a>}-\frac{ p^{(l)} a_{-1}^{(l)} b^{(k)} }{<b,\ a_{-1}>}
-\sum_{r=1}^n  \sum_{j=1}^{n-1} \frac{ a^{(l)} b^{(k)} (\omega^{(r-k)j}+\omega^{(l-r)j} )}{\mu (1-\omega^{j})}
\left(\frac{  b_1^{(r)} a^{(r)} }{<b_1,\ a>}-\frac{  b^{(r)} a_{-1}^{(r)} }{ <b,\ a_{-1}>}\right)\nonumber\\
&&+\sum_{r=1}^n \sum_{j=0}^{n-1} \left( \frac{\mu  a^{(l)} \omega^{(k-r)j} {b^{(r)}}^2 a_{-1}^{(k)}}{(\omega^{j}-\mu^2) <b,\ a_{-1}>}
- \frac{\mu  b_1^{(l)} \omega^{(r-l)j} {a^{(r)}}^2 b^{(k)}}{(\omega^{j}-\mu^2) <b_1,\ a>}\right)\nonumber\\
&=&\frac{ a^{(l)} b_1^{(k)} p^{(k)} }{<b_1,\ a>}-\frac{ p^{(l)} a_{-1}^{(l)} b^{(k)} }{<b,\ a_{-1}>}
- \frac{ a^{(l)} b^{(k)}}{\mu} n  \left(l-k-1 \right) (\cS-1)\frac{  b^{(r)}  a_{-1}^{(r)} }{ <b,\ a_{-1}>}\nonumber\\
&&+\frac{p^{(l)} b^{(k)} b_1^{(l)} a^{(l)} } {b^{(l)} <b_1,\ a> } -\frac{a^{(l)} p^{(k)} b^{(k)} a_{-1}^{(k)}} {a^{(k)} <b,\ a_{-1}>}\label{left}
\end{eqnarray}
To compute the left-hand side of (\ref{At}) for the corresponding entry, we first compute $a_{\tau}$ and $b_{\tau}$ using (\ref{eqps}). 
Substituting it into (\ref{atau}) and (\ref{btau}) in Lemma \ref{lem3}, we have
\begin{eqnarray}\label{attau}
&&a_{\tau}^{(i)}
=\frac{ n a^{(i)}}{\mu}\left(\sum_{r=1}^{i-1} (\cS-1) \frac{ b^{(r)} a_{-1}^{(r)}}{<b,\ a_{-1}>}
-\frac{1} { {p^{(n)}}^2 -1} (\cS-1) \frac{ b^{(n)} a_{-1}^{(n)}}{<b,\ a_{-1}>}\right)
+\frac{p^{(i)} } {b^{(i)} }  (\cS-1) \frac{ b^{(i)} a_{-1}^{(i)}}{<b,\ a_{-1}>}
\end{eqnarray}
and 
\begin{eqnarray}
&&b_{\tau}^{(i)}
=\frac{ n b^{(i)}}{\mu}
\left(\sum_{r=1}^{i-1} (\cS-1) \frac{ b^{(r)} a_{-1}^{(r)}}{<b,\ a_{-1}>}
+\frac{1} { {p^{(n)}}^2 -1} (\cS-1) \frac{ b^{(n)} a_{-1}^{(n)}}{<b,\ a_{-1}>}\right)
+\frac{p^{(i)} } {a^{(i)} }  (\cS-1) \frac{ b^{(i)} a_{-1}^{(i)}}{<b,\ a_{-1}>} \nonumber\\
&&\quad +\frac{ n b^{(i)}}{\mu}
\left(\sum_{r=i}^{n} (\cS-1) \frac{2 b^{(r)} a_{-1}^{(r)}}{<b,\ a_{-1}>}-(\cS-1) \frac{ b^{(i)} a_{-1}^{(i)}}{<b,\ a_{-1}>}
\right)\nonumber\\
&&\quad =\frac{ n b^{(i)}}{\mu}
\left(\sum_{r=1}^{i} -(\cS-1) \frac{ b^{(r)} a_{-1}^{(r)}}{<b,\ a_{-1}>}
+\frac{1} { {p^{(n)}}^2 -1} (\cS-1) \frac{ b^{(n)} a_{-1}^{(n)}}{<b,\ a_{-1}>}\right)
+\frac{p^{(i)} } {a^{(i)} }  (\cS-1) \frac{ b^{(i)} a_{-1}^{(i)}}{<b,\ a_{-1}>} .\label{bttau}
\end{eqnarray}
Thus,
\begin{eqnarray}
&& a_{\tau}^{(l)}b^{(k)}+ a^{(l)} b_{\tau}^{(k)}=
\frac{ n a^{(l)} b^{(k)} }{\mu}
\left(\sum_{r=1}^{l-1} -\sum_{r=1}^{k}\right)(\cS-1) \frac{ b^{(r)} a_{-1}^{(r)}}{<b,\ a_{-1}>}\nonumber\\
&&\qquad+\frac{p^{(l)} b^{(k)} b_1^{(l)} a^{(l)} } {b^{(l)} <b_1,\ a> }  
-\frac{p^{(l)} b^{(k)}  a_{-1}^{(l)}}{<b,\ a_{-1}>}
+\frac{a^{(l)} p^{(k)}  b_1^{(k)} }{<b_1,\ a>}
-\frac{a^{(l)} p^{(k)} b^{(k)} a_{-1}^{(k)}} {a^{(k)} <b,\ a_{-1}>},\label{atbt}
\end{eqnarray}
which equals to (\ref{left}) and we complete the proof.
\end{proof}
The following identity will be used in the proof of Theorem \ref{thm4}.
\begin{Pro}\label{proC}
Given the relations (\ref{consb}), (\ref{coar}), (\ref{indc}) and (\ref{eqps}), the following identity holds:
\begin{eqnarray}
 (\cS-1)\frac{1}{<b,\ a_{-1}>}(a_{-1}^{(i)} b^{(i+1)} + \mu^2 a_{-1}^{(i+1)} b^{(i)}) \exp(\phi^{(i)})
 =\frac{\partial}{\partial \tau} a^{(i)} b^{(i+1)} p^{(i)}\exp(\phi^{(i)}).\label{comm1}
\end{eqnarray}
\end{Pro}
\begin{proof} We prove the identity by direct calculation. The left-hand side equals
 \begin{eqnarray*}
&&\frac{ a^{(i)} b_1^{(i+1)} + \mu^2 a^{(i+1} b_1^{(i)})}{<b_1,\ a>} \exp(\phi_1^{(i)})-\frac{a_{-1}^{(i)} b^{(i+1)} + \mu^2 a_{-1}^{(i+1)} b^{(i)}}{<b,\ a_{-1}>}
\exp(\phi^{(i)})\\
&=& \left( \frac{a^{(i)} b_1^{(i+1)} + \mu^2 a^{(i+1)} b_1^{(i)}}{<b_1,\ a>}\frac{p^{(i)}}{p^{(i+1)}}-
\frac{a_{-1}^{(i)} b^{(i+1)} +\mu^2 a_{-1}^{(i+1)} b^{(i)}}{<b,\ a_{-1}>}\right) \exp(\phi^{(i)}) 
\quad (\frac{\exp(\phi_1^{(i)})}{p^{(i)}}=\frac{\exp(\phi^{(i)})}{p^{(i+1)}})\\
&=&\left( \frac{a^{(i)} p^{(i)}}{a^{(i+1)} p^{(i+1)}}\frac{ a^{(i+1)} b_1^{(i+1)}}{<b_1,\ a>} +  \frac{ b^{(i+1)}  {p^{(i)}}^2}{b^{(i)}} 
\frac{a^{(i)} b_1^{(i)}}{<b_1,\ a>}\right) \exp(\phi^{(i)}) \\
&& -\left(
\frac{ b^{(i+1)}}{b^{(i)}}\frac{a_{-1}^{(i)} b^{(i)}}{<b,\ a_{-1}>} +\frac{a^{(i)} p^{(i)} p^{(i+1)}}{a^{(i+1)} } \frac{ a_{-1}^{(i+1)} b^{(i+1)}}{<b,\ a_{-1}>}\right) \exp(\phi^{(i)}) 
\quad (\mu^2 a^{(i+1)} b^{(i)}= a^{(i)} b^{(i+1)}  p^{(i)} p^{(i+1)}).
 \end{eqnarray*}
We now compute the right-hand side of (\ref{comm1}). First using formula (\ref{atbt}), we have
\begin{eqnarray*}
&& a_{\tau}^{(i)}b^{(i+1)}+ a^{(i)} b_{\tau}^{(i+1)}=\frac{ b^{(i+1)}  } {b^{(i)} p^{(i)} } (\cS-1)
\frac{ b^{(i)}  a_{-1}^{(i)}}{<b,\ a_{-1}>}
+\frac{a^{(i)}} {a^{(i+1)}  p^{(i+1)}}  (\cS-1)
\frac{ b^{(i+1)} a_{-1}^{(i+1)}}{<b,\ a_{-1}>} .
\end{eqnarray*}
So the right-hand side of (\ref{comm1}) is equal to
\begin{eqnarray*}
&&\left(a_{\tau}^{(i)} b^{(i+1)} p^{(i)}+a^{(i)} b_{\tau}^{(i+1)} p^{(i)}+a^{(i)} b^{(i+1)} p_{\tau}^{(i)}
+n a^{(i)} b^{(i+1)} p^{(i)} \phi_{\tau}^{(i)}\right) \exp(\phi^{(i)})\\
&&=\frac{ b^{(i+1)}  } {b^{(i)} } \exp(\phi^{(i)}) (\cS-1)
\frac{ b^{(i)}  a_{-1}^{(i)}}{<b,\ a_{-1}>}
+\frac{a^{(i)} p^{(i)}} {a^{(i+1)}  p^{(i+1)}} \exp(\phi^{(i)}) (\cS-1)
\frac{ b^{(i+1)} a_{-1}^{(i+1)}}{<b,\ a_{-1}>}\\
&&+\frac{ n a^{(i)} b^{(i+1)} p^{(i)} }{\mu}  \exp(\phi^{(i)}) \left(
\frac{ b_1^{(i)} a^{(i)}}{<b_1,\ a>}  -\frac{ b^{(i+1)} a_{-1}^{(i+1)}}{<b,\ a_{-1}>}\right),
\end{eqnarray*}
which is the same as the left-hand side after we use (\ref{reab}) and substitute $\frac{n}{\mu} a^{(i)}  p^{(i)}=\frac{{p^{(i)}}^2 -1}{b^{(i)}}$ for the term containing $b_1$ and
$\frac{n}{\mu} b^{(i+1)} =\frac{{p^{(i+1)}}^2 -1}{a^{(i+1)}  p^{(i+1)}}$ for the term containing $a_{-1}$.
\end{proof}

\section*{Acknowledgements}
The paper is supported by AVM's EPSRC grant EP/I038675/1 and JPW's EPSRC grant EP/I038659/1. All authors gratefully acknowledge the financial support.

\end{document}